\documentclass[11pt]{article}

\usepackage{fullpage,amsmath, amssymb, amsthm,stmaryrd}
\usepackage[T1]{fontenc}
\usepackage[latin9]{inputenc}
\usepackage{color}
\usepackage{verbatim}
\usepackage{prettyref}
\usepackage{refstyle}
\usepackage{float}
\usepackage{amsthm}
\usepackage{amsmath}
\usepackage{amssymb}
\usepackage{setspace}
\usepackage[numbers]{natbib}

\usepackage{a4wide}
\ifx\pdfpageheight\undefined
   \usepackage[dvips,colorlinks=true,linkcolor=blue,citecolor=red,
      urlcolor=green]{hyperref}
   \usepackage[dvips]{graphicx}
   \makeatletter
   \edef\Gin@extensions{\Gin@extensions,.mps}
   
   \makeatother
\fi

\usepackage{algorithm}
 \usepackage{algorithmic}
\floatstyle{ruled}
\newfloat{algorithm}{tbp}{loa}
\floatname{algorithm}{Algorithm}
\usepackage{algorithmic}

\newtheorem{thm}{Theorem}

\newtheorem{defn}{Definition}

\newtheorem{exmp}{Example}
\newtheorem{rem}{Remark}
\newtheorem{lem}{Lemma}

\newcommand{\cdeg}{{\rm cdeg}\,}
\newcommand{\rdeg}{{\rm rdeg}\,}
\newcommand{\diag}{{\rm diag}\,}

\DeclareMathOperator{\lcoeff}{lcoeff}

\DeclareMathOperator{\hermiteDiagonal}{HermiteDiagonal}

\DeclareMathOperator{\colBasis}{ColumnBasis}

\def\mnb{\qopname\relax n{MinimalKernelBasis ~ }}

\title{A fast, deterministic algorithm for computing a Hermite Normal Form of a polynomial matrix}

\author{
Wei Zhou and  George Labahn
         \thanks{
             Cheriton School of Computer Science, University of Waterloo, 
             Waterloo ON, Canada N2L 3G1 \texttt{\{w2zhou,glabahn\}@uwaterloo.ca}
          }
}
\begin{document}

\date{}
\maketitle 
\begin{abstract}
Given a square, nonsingular matrix of univariate polynomials $\mathbf{F} \in \mathbb{K}[x]^{n \times n}$ over a field $\mathbb{K}$, we give a fast, deterministic algorithm for finding the
Hermite normal form of $\mathbf{F}$ with complexity $O^{\sim}\left(n^{\omega}d\right)$ where $d$ is the degree of $\mathbf{F}$. Here soft-$O$ notation is Big-$O$ with log factors removed and $\omega$ is the exponent of matrix multiplication. The method relies of a fast algorithm for determining the diagonal entries of its Hermite normal form, having as cost  $O^{\sim}\left(n^{\omega}s\right)$ operations with $s$ the average of the column degrees of $\mathbf{F}$. 
\end{abstract}

\section{Introduction}

For a given square, nonsingular matrix polynomial $\mathbf{F} \in \mathbb{K}[x]^{n \times n}$ there exists a unimodular matrix $\mathbf{U}$ such that $\mathbf{F} \cdot \mathbf{U} = \mathbf{H}$, a matrix in (column) Hermite normal form. Thus 
$$
\mathbf{H} = \left[ \begin{array}{cccc} h_{11} & & \\ h_{21} & h_{22} & & \\
\vdots & \vdots  & \ddots & \\
h_{n1} & \cdots & \cdots & h_{nn} \end{array} \right]
$$
is a lower triangular matrix where each $h_{ii}$ is monic and deg $h_{ij} < \mbox{ deg } h_{ii}$ for all $j<i$.  Other variations include specifying row rather than column forms (in which case the unimodular matrix multiplies on the left rather than the right) or  upper rather than lower triangular form. The Hermite form was first defined by Hermite in 1851 in the context of triangularizing integer matrices.

There has been considerable work on fast algorithms for Hermite form computation. This includes
$O^{\sim}\left(n^{4}d\right)$ algorithms from Hafner and McCurley \cite{hafner} and Iliopoulos \cite{iliopoulos} which control intermediate size by working modulo the determinant. Hafner and McCurley \cite{hafner},  Storjohann and Labahn \cite{storjohann-labahn96} and Villard \cite{villard96} gave new algorithms which reduced the cost to $O^{\sim}\left(n^{\omega + 1}d\right)$ operations, with $2 < \omega < 3$ being the exponent of matrix multiplication. The second named worked with integer matrices but the results carried over directly to polynomial matrices. Mulders and Storjohann \cite{mulders-storjohann:2003} gave an iterative algorithm having complexity $O\left(n^{3}d^2\right)$, thus reducing the exponent of $n$ at the cost of increasing the exponent of the degree $d$. 

During the past decade the goal has been to give an algorithm that computes the Hermite form in the time it takes to multiply two polynomial matrices having the same size $n$ and degree $d$ as the input matrix, namely at a cost $O^{\sim}\left(n^{\omega}d\right)$. Such algorithms already exist for a number of other polynomial matrix problems. This includes probabalistic algorithms for linear solving \cite{mulders-storjohann:2003}, row reduction \cite{Giorgi2003} and polynomial matrix inversion \cite{jeannerod-villard:05} and later deterministic algorithms for linear solving and row reduction \cite{GSSV2012}. In the case of Hermite normal form computation Gupta and Storjohann \cite{GS2011} gave a Las Vegas randomized algorithm with expected running time of $O^{\sim}\left(n^{3}d\right)$. Their algorithm was the first to be both softly cubic in $n$ and linear in $d$.

One natural technique for finding a Hermite form is to first determine a triangular form and to then reduce the lower triangular elements using the diagonals. The problem with this is that the best a-priori bounds for the degrees of the unimodular multiplier $\mathbf{U}$ can become too large for efficient computation (since these bounds are determined from $\frac{1}{\det \mathbf{F}} \mbox{ adj } (\mathbf{F} ) \cdot \mathbf{H}$). On the other hand simply looking for bounds on $\mathbf{H}$ has a similar problem since the best known a-priori bound for the $i$-th column is
$i d$ and hence the sum of these degree bounds is $O(n^2 d)$, a factor of $n$ larger than the
actual sum $\mbox{ deg det}(\mathbf{H}) = n d$.  Gupta and Storjohann make use of the Smith
normal form of $\mathbf{F}$ in order to obtain accurate bounds for the degrees of the diagonal entries (and hence the degrees of the columns) of $\mathbf{H}$.  That combined with some
additional partial information on one of the right multipliers of this Smith form are then used
to find $\mathbf{H}$.

In this paper we give a deterministic Hermite norml form algorithm having complexity 
$O^{\sim}\left(n^{3}d\right)$. As with Gupta and Storjohan ours is a two step process. 
We first determine the diagonal elements of $\mathbf{H}$ and then secondly find the remaining elements having reduced degrees. Our approach is to make use of fast, deterministic methods for shifted  minimal kernel basis and column basis computation to find the diagonal entries. We do this without the need for finding the associated unimodular multiplier. We do this with a cost $O^{\sim}\left(n^{\omega}s\right)$ field operations where $s$ is the average of the column degrees of $\mathbf{F}$.  The remaining entries are then determined making use of a second type of fast shifted minimal kernel basis computation with special care required to reduce the computation to one having small degrees. The use of shifted minimal kernel bases for matrix normal form computation was previously used in \cite{BLV:1999,BLV:jsc06} in order to obtain efficient algorithms in the case where intermediate coefficient growth is a concern. 

The remainder of this paper is organized as follows. In the next section we give preliminary information for shifted degrees, kernel and column bases of polynomial matrices. Section 3 then contains the algorithm for finding the diagonal elements of a Hermite form with the following section giving the details of the fast algorithm for the entire Hermite normal form computation.  The paper ends with a conclusion and topics for future research.

\section{Preliminaries}

In this section we first describe the notations used in this paper,
and then give the basic definitions and properties of {\em shifted
degree}, {\em kernel basis} and {\em column basis} for a matrix
of polynomials. These will be the building blocks used in our algorithm.

\subsection{Shifted Degrees}

Our methods makes use of the concept of {\em shifted} degrees of polynomial matrices \cite{BLV:1999}, basically shifting the importance of the degrees in some of the rows of 
a basis. For a column vector $\mathbf{p}=\left[p_{1},\dots,p_{n}\right]^{T}$ of 
univariate polynomials over a field $\mathbb{K}$, its column degree, denoted by $\cdeg\mathbf{p}$,
is the maximum of the degrees of the entries of $\mathbf{p}$, that is, 
\[
\cdeg~\mathbf{p}=\max_{1\le i\le n}\deg p_{i}.
\]
The \emph{shifted column degree} generalizes this standard column degree by taking the maximum after shifting the degrees by a given integer vector that is known as a \emph{shift}. More specifically, the shifted column degree of $\mathbf{p}$ with respect to a shift
$\vec{s}=\left[s_{1},\dots,s_{n}\right]\in\mathbb{Z}^{n}$, or the
\emph{$\vec{s}$-column degree} of $\mathbf{p}$ is 
\[
\cdeg_{\vec{s}}~\mathbf{p}=\max_{1\le i\le n}[\deg p_{i}+s_{i}]=\deg(x^{\vec{s}}\cdot\mathbf{p}),
\]
where 
\[
x^{\vec{s}}=\diag\left(x^{s_{1}},x^{s_{2}},\dots,x^{s_{n}}\right) ~.
\]
For a matrix $\mathbf{P}$, we use $\cdeg\mathbf{P}$ and $\cdeg_{\vec{s}}\mathbf{P}$
to denote respectively the list of its column degrees and the list
of its shifted $\vec{s}$-column degrees. When $\vec{s}=\left[0,\dots,0\right]$,
the shifted column degree specializes to the standard column degree.
Similarly, $\cdeg_{- \vec{s}}\mathbf{P} \leq 0$ is equivalent to deg $p_{ij} \leq s_i$ for all $i$ and $j$, that is,  $\vec s$ bounds the row degrees of $\mathbf{P}$.

The shifted row degree of a row vector \textbf{$\mathbf{q}=\left[q_{1},\dots,q_{n}\right]$}
is defined similarly as 
\[
\rdeg_{\vec{s}}\mathbf{q}=\max_{1\le i\le n}[\deg q_{i}+s_{i}]=\deg(\mathbf{q}\cdot x^{\vec{s}}).
\]
Shifted degrees have been used previously in polynomial matrix
computations and in generalizations of some matrix normal forms \cite{BLV:jsc06}.
The shifted column degree is equivalent to the notion of \emph{defect}
commonly used in the literature. 

Along with shifted degrees we also make use of the notion of a matrix polynomial being column (or row) reduced. A matrix polynomial $\mathbf{F}$ is column reduced if the leading column coefficient matrix, that is the matrix 
$$
[ \mbox{coeff}( f_{ij}, x, d_j ) ]_{1\leq i, j \leq n}, \mbox{ with } \vec d = \mbox{cdeg } \mathbf{F}, 
$$
has full rank. A matrix polynomial $\mathbf{F}$ is $\vec s$-column reduced if $x^{\vec s} \mathbf{F} $ is column reduced. A similar concept exists for being shifted row reduced.

The usefulness of the shifted degrees can be seen from their applications in polynomial matrix computation problems \cite{ZL2012,za2012}. One of its uses is illustrated by the following lemma from \cite[Chapter 2]{zhou:phd2012}, which can be viewed as a stronger version of the predictable-degree property \cite[page 387]{kailath:1980}. For completeness we also include the proof.
\begin{lem}\label{lem:predictableDegree}
Let $\mathbf{A}\in\mathbb{K}\left[x\right]^{m\times n}$ be a $\vec{u}$-column reduced matrix with no zero columns and with $\cdeg_{\vec{u}}\mathbf{A}=\vec{v}$. Then a matrix $\mathbf{B}\in\mathbb{K}\left[x\right]^{n\times k}$ has $\vec{v}$-column degrees $\cdeg_{\vec{v}}\mathbf{B}=\vec{w}$ if and only if $\cdeg_{\vec{u}}\left(\mathbf{A}\mathbf{B}\right)=\vec{w}$. 
\end{lem}
\begin{proof}
Being $\vec{u}$-column reduced with $\cdeg_{\vec{u}}\mathbf{A}=\vec{v}$ is equivalent to the leading coefficient matrix of $x^{\vec u} \cdot \mathbf{A} \cdot x^{- \vec{v}}$ having linearly independent columns. The leading coefficient matrix of $x^{\vec v} \cdot \mathbf{B} \cdot x^{- \vec{w}}$ has no zero column if and only if the leading coefficient matrix of 
$$
x^{\vec u} \cdot \mathbf{A B} \cdot x^{- \vec{w}}
= x^{\vec u} \cdot \mathbf{A} \cdot x^{- \vec{v}} x^{\vec v} \cdot \mathbf{B} \cdot x^{- \vec{w}}
$$
has independent columns.
That is, $x^{\vec v} \cdot \mathbf{B} \cdot x^{- \vec{w}}$ has column degree $\vec 0$ if and only if $x^{\vec u} \cdot \mathbf{A B} \cdot x^{- \vec{w}}$ has column degree $\vec 0$.
\end{proof}
An essential fact needed in this paper, also based on the use of shifted degrees, is the efficient multiplication of matrices with unbalanced degrees \cite[Theorem 3.7]{za2012}. 
\begin{thm}\label{thm:multiplyUnbalancedMatrices} 
Let $\mathbf{A}\in\mathbb{K}\left[x\right]^{m\times n}$
with $m\le n$, $\vec{s}\in\mathbb{Z}^{n}$ a shift with entries bounding
the column degrees of $\mathbf{A}$ and $\xi$, a bound on the sum
of the entries of $\vec{s}$. Let $\mathbf{B}\in\mathbb{K}\left[x\right]^{n\times k}$
with $k\in O\left(m\right)$ and the sum $\theta$ of its $\vec{s}$-column
degrees satisfying $\theta\in O\left(\xi\right)$. Then we can multiply
$\mathbf{A}$ and $\mathbf{B}$ with a cost of $O^{\sim}(n^{2}m^{\omega-2}s)\subset O^{\sim}(n^{\omega}s)$,
where $s=\xi/n$ is the average of the entries of $\vec{s}$. 
\end{thm}

\subsection{Kernel Bases and Column Bases}

The kernel of $\mathbf{F}\in\mathbb{K}\left[x\right]^{m\times n}$
is the $\mathbb{F}\left[x\right]$-module 
\[
\left\{ \mathbf{p}\in\mathbb{K}\left[x\right]^{n}~|~\mathbf{F}\mathbf{p}=0\right\} 
\]
with a kernel basis of $\mathbf{F}$ being a basis of this module. Formally, we have:

\begin{defn}
\label{def:kernelBasis}Given $\mathbf{F}\in\mathbb{K}\left[x\right]^{m\times n}$,
a polynomial matrix $\mathbf{N}\in\mathbb{K}\left[x\right]^{n\times k}$
is a (right) kernel basis of $\mathbf{F}$ if the following properties
hold: 
\begin{enumerate}
\item $\mathbf{N}$ is full-rank. 
\item $\mathbf{N}$ satisfies $\mathbf{F}\cdot\mathbf{N}=0$. 
\item Any $\mathbf{q}\in\mathbb{K}\left[x\right]^{n}$ satisfying $\mathbf{F}\mathbf{q}=0$
can be expressed as a linear combination of the columns of $\mathbf{N}$,
that is, there exists some polynomial vector $\mathbf{p}$ such that
$\mathbf{q}=\mathbf{N}\mathbf{p}$. 
\end{enumerate}
\end{defn}
\noindent
It is easy to show that any pair of kernel bases $\mathbf{N}$ and $\mathbf{M}$ of $\mathbf{F}$
are unimodularly equivalent.

A $\vec{s}$-minimal kernel basis of $\mathbf{F}$ is just a kernel basis that is $\vec{s}$-column reduced. 
\begin{defn}
Given $\mathbf{F}\in\mathbb{K}\left[x\right]^{m\times n}$, a polynomial
matrix $\mathbf{N}\in\mathbb{K}\left[x\right]^{n\times k}$ is a $\vec{s}$-minimal
(right) kernel basis of $\mathbf{F}$ if\textbf{ $\mathbf{N}$} is
a kernel basis of $\mathbf{F}$ and $\mathbf{N}$ is $\vec{s}$-column
reduced. We also call a $\vec{s}$-minimal (right) kernel basis of
$\mathbf{F}$ a $\left(\mathbf{F},\vec{s}\right)$-kernel basis. 

\end{defn}


A column basis of $\mathbf{F}$ is a basis for the $\mathbb{K}\left[x\right]$-module 
\[
\left\{ \mathbf{F}\mathbf{p}~|~\mathbf{p}\in\mathbb{K}\left[x\right]^{n}~\right\} ~.
\]
Such a basis can be represented as a full rank matrix 
$\mathbf{T}\in\mathbb{K}\left[x\right]^{m\times r}$
whose columns are the basis elements. A column basis is not
unique and indeed any column basis right multiplied by a unimodular
polynomial matrix gives another column basis.

The cost of kernel basis computation is given in \cite{za2012} while the cost of column basis computation is given in \cite{za2013}. In both cases they make heavy use of fast methods for
order bases (often also referred to as minimal approximant bases) \cite{BeLa94,Giorgi2003,ZL2012}.

\begin{thm}\label{thm:costGeneral} 
Let $\mathbf{F}\in\mathbb{K}\left[x\right]^{m\times n}$ with $\vec{s}=\cdeg\mathbf{F}$. Then a $\left(\mathbf{F},\vec{s}\right)$-kernel basis can be computed with a cost of $O^{\sim}\left(n^{\omega}s\right)$ field operations where $s=\sum\vec{s}/n$ is the average
column degree of $\mathbf{F}$.
\end{thm}

\begin{thm}\label{thm:fastcolbasis}
There exists a fast, deterministic algorithm for the computation of a column basis of a matrix polynomial $\mathbf{F}$ having complexity $O^{\sim}\left(nm^{\omega-1}s\right)$
field operations in $\mathbb{K}$  with $s$ being the average column degree of $\mathbf{F}$.
\end{thm}

\subsection{Example}

\begin{exmp}\label{example1}
Let 
$$
\mathbf{F} = \left[ \begin {array}{rcrcr} x&-{x}^{3}&-2\,{x}^{4}& 2x&-{x}^{2}
\\ \noalign{\medskip}1&-1&-2\,x&2&-x\\ \noalign{\medskip}-3&3\,{x}^{2}
+x&2\,{x}^{2}&-\,{x}^{4} + 1&3\,x\end {array} \right]
$$
be a $3 \times 5$ matrix over $\mathbb{Z}_7[x]$ having column degree $\vec s = (1,3,4, 4, 2)$. Then a column space, $\mathbf{G}$, and a kernel basis, $\mathbf{N}$, of $\mathbf{F}$ are given by 
$$
\mathbf{G} = \left[ \begin {array}{rcr} x&-{x}^{3}&-2\,{x}^{4}
\\ \noalign{\medskip}1&-1&-2\,x\\ \noalign{\medskip}-3&3\,{x}^{2}+
x&2\,{x}^{2}\end {array} \right] ~~
\mbox { and } ~~\mathbf{N} := \left[ \begin {array}{rc} -1&x\\ \noalign{\medskip}-{x}^{2}&0
\\ \noalign{\medskip}- 3\,x&0\\ \noalign{\medskip}-3&0
\\ \noalign{\medskip}0&1\end {array} \right] ~.
$$
For example, if $\{\mathbf{g}_i\}_{i=1, ... , 5}$ denote the columns of $\mathbf{G}$ then column $4$ of $\mathbf{F}$ - denoted by $\mathbf{f}_4$ -  is given by
$$
\mathbf{f}_4  =  - 2 ~\mathbf{g}_1  - 2 x^2 ~ \mathbf{g}_2 + x ~\mathbf{g}_3  + 2 ~\mathbf{g}_4.
$$
Here $\cdeg_{\vec s} \mathbf{N} = (5, 2)$ with shifted leading coefficient matrix
$$
\mbox{lcoeff}_{\vec s} (\mathbf{N} ) = \left[ \begin{array}{rc} 0 & 1 \\ - 1 & 0 \\ - 3 & 0 \\ 0 & 0 \\ 0 & 1 \end{array} \right] .
$$
Since $\mbox{ lcoeff}_{\vec s}( \mathbf{N})$ has full rank we have that $\mathbf{N}$ is a $\vec s$-minimal kernel basis.
\qed
\end{exmp}

\section{\label{sec:diagonals}Determining the Diagonal Entries of a Hermite Normal Form}

In this section we first show how to determine only the diagonal entries of the Hermite normal form of a nonsingular input matrix $\mathbf{F}\in\mathbb{K}\left[x\right]^{n\times n}$ with
$\mathbf{F}$ having column degrees $\vec{s}$. The computation makes use of fast kernel and column basis computation.

Consider unimodularly transforming $\mathbf{F}$ to 
\begin{equation}
\mathbf{F} \cdot \mathbf{U}=\mathbf{G}=\begin{bmatrix}\mathbf{G}_{1} & 0\\
* & \mathbf{G}_{2}
\end{bmatrix} .  \label{eq:step1HermiteDiagonal}
\end{equation}
After this unimodular transformation, the elimination of the top right block of $\mathbf{G}$, the matrix is now closer to being in Hermite normal form. Applying this procedure recursively
to $\mathbf{G}_{1}$ and $\mathbf{G}_{2}$, until the matrices reach dimension $1$, gives the diagonal entries of the Hermite
normal form of $\mathbf{F}$.

While such a procedure can be used to correctly compute the diagonal entries of the Hermite 
normal form of $\mathbf{F}$, a major problem is that the degree of the unimodular multiplier $\mathbf{U}$ can be too large for efficient computation. Our approach is to make use of
fast kernel and column basis methods to efficiently compute only $\mathbf{G}_{1}$ and $\mathbf{G}_{2}$ and so avoid computing $\mathbf{U}$.

Partition 
$\mathbf{F}=\begin{bmatrix}\mathbf{F}_{u}\\\mathbf{F}_{d} \end{bmatrix}$, with 
$\mathbf{F}_{u}$ and $\mathbf{F}_{d}$ consisting of the upper $\lceil n/2\rceil$ and lower $\left\lfloor n/2\right\rfloor $ rows of $\mathbf{F}$, respectively.
Then both upper and lower parts are of full-rank since $\mathbf{F}$ is assumed to be nonsingular. By partitioning $\mathbf{U}=\begin{bmatrix}\mathbf{U}_{\ell} ~, & \mathbf{U}_{r}\end{bmatrix}$, where the column dimension of $\mathbf{U}_{\ell}$ matches the row dimension of $\mathbf{F}_{u}$, then $\mathbf{F} \cdot \mathbf{U}= \mathbf{G}$ becomes
\[
\begin{bmatrix}\mathbf{F}_{u}\\
\mathbf{F}_{d}
\end{bmatrix}\begin{bmatrix}\mathbf{U}_{\ell} & \mathbf{U}_{r}\end{bmatrix}=\begin{bmatrix}\mathbf{G}_{1} & 0\\
* & \mathbf{G}_{2}
\end{bmatrix}.
\]
Notice that the matrix $\mathbf{G}_{1}$ is nonsingular and is therefore a column basis of $\mathbf{F}_{u}$. As such this can be efficiently computed
as mentioned in Theorem \ref{thm:fastcolbasis}. In order to compute $\mathbf{G}_{2}=\mathbf{F}_{d}\mathbf{U}_{r}$,
notice that the matrix $\mathbf{U}_{r}$ is a right kernel basis of $\mathbf{F}_u$, which makes the top right block of $\mathbf{G}$ zero. 

The following lemma states that the kernel basis $\mathbf{U}_{r}$ can be replaced by any other kernel basis of $\mathbf{F}_u$ to give another unimodular matrix that also works. 
\begin{lem}
\label{lem:oneStepHermiteDiagonal}
Partition
$\mathbf{F}=\begin{bmatrix}\mathbf{F}_{u}\\
\mathbf{F}_{d} \end{bmatrix}$ and suppose $\mathbf{G}_{1}$ is a column basis of 
$\mathbf{F}_{u}$ and $\mathbf{N}$ a kernel basis of $\mathbf{F}_{u}$. Then there is
a unimodular matrix $\mathbf{U}=\left[ ~ *~ ,~ \mathbf{N}\right]$ such that
\[
\mathbf{F}\mathbf{U}=\begin{bmatrix}\mathbf{G}_{1} & 0 \\
* & \mathbf{G}_{2}
\end{bmatrix},
\]
where $\mathbf{G}_{2}=\mathbf{F}_{d}\mathbf{N}$.  If $\mathbf{F}$
is square nonsingular, then $\mathbf{G}_{1}$ and $\mathbf{G}_{2}$
are also square nonsingular.
\end{lem}
\begin{proof} This follows from Lemma 3.1 in \cite{za2013} \end{proof}

Note that we do not compute the blocks represented by the symbol $*$,
which may have very large degrees and cannot be computed efficiently.

\prettyref{lem:oneStepHermiteDiagonal} allows us to determine $\mathbf{G}_{1}$
and $\mathbf{G}_{2}$ independently without computing the unimodular
matrix. 
%
This procedure for computing the diagonal entries gives Algorithm 
\prettyref{alg:hermiteDiagonal}. Formally 
the cost of this algorithm is given in Theorem \ref{thm:diagonalCost}.

\begin{algorithm}[t]
\caption{$\hermiteDiagonal(\mathbf{F})$}
\label{alg:hermiteDiagonal}

\begin{algorithmic}[1]
\REQUIRE{$\mathbf{F}\in\mathbb{K}\left[x\right]^{n\times n}$ is nonsingular.
}

\ENSURE{$\mathbf{d}\in\mathbb{K}\left[x\right]^{n}$ a list of diagonal entries
of the Hermite normal form of $\mathbf{F}$.}

\STATE{Partition $\mathbf{F} := \begin{bmatrix}\mathbf{F}_{u}\\
\mathbf{F}_{d}
\end{bmatrix}$, with $\mathbf{F}_{u}$ consists of the top $\left\lceil n/2\right\rceil $
rows of $\mathbf{F}$;}

\STATE{\textbf{if }$n=1$ \textbf{then} \textbf{return} $\mathbf{F}$; \textbf{endif};}

\STATE{$\mathbf{G}_{1}:=\colBasis(\mathbf{F}_{u})$;}

\STATE{$\mathbf{N}:=\mnb(\mathbf{F}_{u},\cdeg\mathbf{F})$;}

\STATE{$\mathbf{G}_{2}:=\mathbf{F}_{d}\mathbf{N}$;}

\STATE{$\mathbf{d}_{1}:=\hermiteDiagonal(\mathbf{G}_{1}) $
       $~~~~\mathbf{d}_{2}:=\hermiteDiagonal(\mathbf{G}_{2});$}

\STATE{\textbf{return} $\left[\mathbf{d}_{1},\mathbf{d}_{2}\right]$;}
\end{algorithmic}
\end{algorithm}

\begin{exmp}\label{example2}
Let
$$
\mathbf{F} = \left[ \begin {array}{rcccc} x&-{x}^{3}&-2\,{x}^{4}&2x&-{x}^{2}
\\ \noalign{\medskip}1&-1&-2\,x&2&-x\\ \noalign{\medskip}-3&3\,{x}^{2}
+x&2\,{x}^{2}&-\,{x}^{4} + 1&3\,x\\ \noalign{\medskip}0&1&{x}^{2}+2\,x-2
&\,{x}^{3}+2 x-2&0\\ \noalign{\medskip}1&-{x}^{2}+2&-2\,{x}^{3}-3\,x+3
&2 x+ 2&0\end {array} \right]
$$
working over $\mathbb{Z}_7[x]$. If $\mathbf{F}_u$ denotes the top three rows of $\mathbf{F}$ then a column basis $\mathbf{G}_1$ and kernel basis $\mathbf{N}$ were given in Example \ref{example1}.
If $\mathbf{F}_d$ denotes the bottom $2$ rows of $\mathbf{F}$, then this gives diagonal blocks $\mathbf{G}_1$ and $\mathbf{G}_2 = \mathbf{F}_d \cdot \mathbf{N}$ as
$$
\mathbf{G}_1 = \left[ \begin {array}{rcr} x&-{x}^{3}&-2\,{x}^{4}
\\ \noalign{\medskip}1&-1&-2\,x\\ \noalign{\medskip}-3&3\,{x}^{2}+
x&2\,{x}^{2}\end {array} \right] 
\mbox { and } ~~
\mathbf{G}_2 =  \left[ \begin {array}{cc} {x}^{3}-1&0\\ \noalign{\medskip}-x&x
\end {array} \right] ~
$$
Recursively computing with $\mathbf{G}_1$ gives column space and nullspace for the top $2$ rows, $\mathbf{G_U}_1$ as
$$
\tilde{ \mathbf{G} }_1 = \left[ \begin{array}{cc} x & 0 \\ 1 & x^2 - 1 \end{array} \right ] ~~
\mbox{ and } ~~ \tilde{\mathbf{N}} = \left[ \begin{array}{r}  0 \\ -2 x \\1  \end{array} \right ] .
$$
This in turn gives $\tilde{ \mathbf{G} }_2 =  { \mathbf{G_d} }_2 \cdot \tilde{\mathbf{N}} = [ x^3 ]$ which gives the two diagonal blocks from $\mathbf{G}_1$. As $\tilde{ \mathbf{G} }_1$ is triangular we have triangularized $\mathbf{G}_1$. Similarly as $\mathbf{G}_2$ is already in triangular form we do not need to do any extra work. As a result we have that  $\mathbf{F}$ is unimodularly equivalent to 
$$
\left[ \begin{array}{ccccc} x&&&&\\ 
  \noalign{\medskip} * &{x}^{2}-1&&&\\ 
  \noalign{\medskip}*&*&{x}^{3}&&\\ 
  \noalign{\medskip}* & * & * &{x}^{3}-1&\\ 
  \noalign{\medskip} *&*&*&*&x
\end{array} \right]
$$
giving the diagonal elements of the Hermite form of $\mathbf{F}$.
\qed
\end{exmp}


\begin{thm}\label{thm:diagonalCost}
Algorithm \prettyref{alg:hermiteDiagonal} 
costs $O^{\sim}\left(n^{\omega}s\right)$ field operations to compute the diagonal entries 
for the Hermite normal form of a nonsingular matrix $\mathbf{F}\in\mathbb{K}\left[x\right]^{n\times n}$, where
$s$ is the average column degree of $\mathbf{F}$. 
\end{thm}
\begin{proof}
The three main operations are computing a column basis of $\mathbf{F}_{u}$,
computing a kernel basis $\mathbf{N}$ of $\mathbf{F}_{u}$, and multiplying the
matrices $\mathbf{F}_{d}\mathbf{N}$.  Set $\xi = \sum\vec{s}$, a scalar used to measure size for our problem.

For the column basis computation, by Theorem \ref{thm:fastcolbasis} (see also \cite[Theorem 5.6]{za2013}) we know that a column basis $\mathbf{G}_{1}$ of $\mathbf{F}_{u}$ can be computed with a cost of $O^{\sim}\left(n^{\omega}s\right).$ By \cite[Lemma 5.1]{za2013} the column degrees
of the computed column basis $\mathbf{G}_{1}$ are also bounded by the original column degrees $\vec{s}$. Similarly, from Theorem \ref{thm:costGeneral} (see also \cite[Theorem 4.1]{za2012}),
computing a $\vec{s}$-minimal kernel basis $\mathbf{N}$ of $\mathbf{F}_{u}$ also costs $O^{\sim}\left(n^{\omega}s\right)$ operations.
 
By Theorem 3.4 of \cite{za2012} we also know that the sum of the $\vec{s}$-column
degrees of the output kernel basis $\mathbf{N}$ is bounded by $\xi$. For the matrix multiplication $\mathbf{F}_{d}\mathbf{N}$, we have that the sum of the column degrees of $\mathbf{F}_{d}$ and the sum of the $\vec{s}$-column degrees of $\mathbf{N}$ are both bounded by $\xi$. Therefore Theorem \ref{thm:multiplyUnbalancedMatrices} (see also \cite[Theorem 3.7]{za2012}) applies and the multiplication can be done with a cost of $O^{\sim}\left(n^{\omega}s\right)$.

If we let the cost of Algorithm \prettyref{alg:hermiteDiagonal} be $g(n)$ for a input matrix of dimension $n$ then
\begin{eqnarray*}
g(n) & \in & O^{\sim}(n^{\omega} s)+g(\left\lceil n/2\right\rceil)+g(\left\lfloor n/2\right\rfloor ).
\end{eqnarray*}
As $s$ depends on $n$ we use $O^{\sim}\left(n^{\omega}s\right) = O^{\sim}\left(n^{\omega-1} \xi \right)$ with $\xi$ not depending on $n$. Then we solve the recurrence relation as
\begin{eqnarray*}
g(n) & \in & O^{\sim}(n^{\omega-1} \xi)+g(\left\lceil n/2\right\rceil)+g(\left\lfloor n/2\right\rfloor )\\
 & \in & O^{\sim}(n^{\omega-1} \xi )+2g(\left\lceil n/2\right\rceil )\\
 & \in & O^{\sim}(n^{\omega-1} \xi) = O^{\sim}(n^{\omega} s).
\end{eqnarray*}
\end{proof}

\section{\label{sec:hermite} Hermite Normal Form}

In \prettyref{sec:diagonals}, we have shown how to efficiently determine the diagonal entries
of the Hermite normal form of a nonsingular input matrix $\mathbf{F}\in\mathbb{K}\left[x\right]^{n\times n}$. In this section we show how to determine the remaining entries for the complete Hermite form $\mathbf{H}$ of $\mathbf{F}$. Our approach is similar to that used in the randomized algorithm of \cite{GS2011} but does not use the Smith normal form. Our algorithm has the advantage of being both efficient and deterministic.

\subsection{Hermite Form via Minimal Kernel Bases}\label{subsec:1}

For simplicity, let us assume that $\mathbf{F}$ is already column reduced, something which we can do with complexity $O^{\sim}\left(n^{\omega}s\right)$ using the column basis algorithm of \cite{za2013}. Assume that $\mathbf{F}$ has column degrees $\vec{d}~=~\left[d_{1},\dots,d_{n}\right]$ and that $\vec{s}=\left[s_{1},\dots,s_{n}\right]$ are the degrees of the diagonal entries of the Hermite form $\mathbf{H}$. Let $d_{\max} =\max_{1\leq i \leq n} {d}_i$ and $s_{\max} =\max_{1\leq i \leq n} {s}_i$ and set 
$\vec{u}=\left[s_{\max} , \dots, s_{\max} \right]$, a vector  with
$n$ entries. The following lemma implies that we can obtain the Hermite normal form of $\mathbf{F}$ from a $\left[-\vec{u},-\vec{s}\right]$-minimal kernel basis of $\left[\mathbf{F}, - I \right]  \in \mathbb{K}\left[x\right]^{n\times 2n} $.
\begin{lem}
Suppose $\mathbf{N} \in \mathbb{K}\left[x\right]^{2n\times n}$ is a $\left[-\vec{u},-\vec{s}\right]$-minimal kernel basis of $\left[\mathbf{F},-I\right]$, partitioned 
as $\mathbf{N} := \begin{bmatrix}\mathbf{N}_u\\
\mathbf{N}_d
\end{bmatrix}$  
where each block is $n\times n$ square. Then $\mathbf{N}_u$ is unimodular while $\mathbf{N}_d$ has row degrees $\vec{s}$ and is unimodularly equivalent to the Hermite normal form $\mathbf{H}$ of $\mathbf{F}$.
\end{lem}
\begin{proof}
Let $\mathbf{U}$ be the unimodular matrix satisfying $\mathbf{F} \cdot \mathbf{U}=\mathbf{H}$. By the predictable-degree property (\ref{lem:predictableDegree}), $\mathbf{U}$ has $\vec{d}$-column degrees bounded by $s_{\max}$, that is, $\cdeg_{-\vec{u}}\mathbf{U}~\le~0$. Letting $\mathbf{M} = \begin{bmatrix}\mathbf{U}\\
\mathbf{H}
\end{bmatrix} \in \mathbb{K}[x]^{2n \times n}$ 
we have that $\mathbf{M}$  is a kernel basis of $\left[\mathbf{F},-I\right]$ having shifted $[ - \vec{u},-\vec{s}]$ column degree $0$. Thus $\mathbf{M}$ is  $\left[-\vec{u},-\vec{s}\right]$-column reduced and hence is a $\left[-\vec{u},-\vec{s}\right]$-minimal
kernel basis of $\left[\mathbf{F},-I\right]$. 

The $\left[-\vec{u},-\vec{s}\right]$-minimal kernel basis  $\mathbf{N} $
of $\left[\mathbf{F},-I\right]$ is then unimodularly equivalent to $\mathbf{M}$. Thus the matrix
$\mathbf{N}_u$, making up the  upper $n$ rows, is unimodular and the matrix $\mathbf{N}_d$, consisting of the lower $n$ rows, is unimodularly equivalent to $\mathbf{H}$. Minimality also ensures that
 $\cdeg_{\left[-\vec{u},-\vec{s}\right]}\begin{bmatrix}\mathbf{N}_u\\
\mathbf{N}_d
\end{bmatrix}=\cdeg_{\left[-\vec{u},-\vec{s}\right]}\begin{bmatrix}\mathbf{U}\\
\mathbf{H}
\end{bmatrix}=0$, 
implying $\cdeg_{-\vec{s}}\mathbf{N}_d \le0=\cdeg_{-\vec{s}}\mathbf{H}$.
Similarly, the minimality of $\mathbf{H}$ then ensures that $\cdeg_{-\vec{s}}\mathbf{N}_d =\cdeg_{-\vec{s}}\mathbf{H}=0$,
or equivalently, $\rdeg\mathbf{N}_d=\rdeg\mathbf{H}=\vec{s}$.
\end{proof}
Knowing that the bottom $n$ rows $\mathbf{N}_d$ have the same row degrees as $\mathbf{H}$
and is unimodularly equivalent with $\mathbf{H}$ implies that the Hermite form
$\mathbf{H}$ can then be obtained from $\mathbf{N}_d$ using Lemma
8 from \cite{GS2011}. We restate this as follows:
\begin{lem}\label{lem:recoverH}
Let $\mathbf{H}\in\mathbb{K}[x]^{n\times n}$ be a column basis of a matrix $\mathbf{A}\in\mathbb{K}\left[x\right]^{n\times k}$ having the same row degrees as $\mathbf{A}$.
If $U\in\mathbb{K}^{n\times n}$ is the matrix putting $\lcoeff\left(x^{-\vec{s}}\mathbf{A}\right)$ into reduced column echelon form and if $\mathbf{H}$ is in Hermite normal form 
then $\mathbf{H}$ is the principal $n\times n$ submatrix of $\mathbf{A} \cdot U$.
\end{lem}
\begin{proof}
This follows from the fact that $\mathbf{H}$ and $\mathbf{A}$ all
have uniform $-\vec{s}$ column degrees $0$. This allows their relationship
to be completely determined by the leading coefficient matrices of both $\left(x^{-\vec{s}}\mathbf{A}\right)$
and $\left(x^{-\vec{s}}\mathbf{H}\right)$, respectively.
\end{proof}

\begin{exmp}\label{example3}
Let
$$
\mathbf{F} = \left[ \begin {array}{ccc} 2\,{x}^{3}-2\,{x}^{2}+3\,x-3&-2\,{x}^{3}+2
\,{x}^{2}&-2\,x+2\\ \noalign{\medskip}{x}^{3}+3\,{x}^{2}-x+2&-2\,{x}^{
2}+x+1&-{x}^{3}-{x}^{2}-2\\ \noalign{\medskip}-3\,{x}^{3}+x-1&-x&1
\end {array} \right] 
$$
be a $3 \times 3$ nonsingular matrix over $\mathbb{Z}_7[x]$. As the leading column coefficient matrix is nonsingular, $\mathbf{F}$ is column reduced.  Using the method of Section \ref{sec:diagonals} one can determine that the diagonal elements of the Hermite form of $\mathbf{F}$ are $x-1, x+1, x^7 +1$, respectively.  Using \cite{za2012} one computes a kernel basis for $[\mathbf{F}, - \mathbf{I}]$ as
$$
\mathbf{N} = \left[ \begin {array}{cccccc} 1&1&1&x-1&1&-3\,{x}^{3}
\\ \noalign{\medskip}2\,{x}^{2}+1&2\,{x}^{2}+2&2\,{x}^{2}+1&x-1&x+2&{x
}^{5}-3\,{x}^{3}-x\\ \noalign{\medskip}2\,{x}^{4}&2\,{x}^{4}+{x}^{2}&2
\,{x}^{4}+1&-2\,x+2&-2&{x}^{7}-{x}^{3}+1\end {array} \right]^t .
$$
With shift $\vec w = [-7, -7, -7, -1, -1, -7]$ we have that
$$
\mbox{ lcoeff}_{\vec w} (\mathbf{N} ) = \left[ \begin {array}{cccrcc} 0 & 0 & 0 & 1 & 0 & 0\\ \noalign{\medskip} 0 & 0 & 0 & 1 & 1 & 0\\ \noalign{\medskip} 0 & 0 & 0 & -2 & 0 & 1 \end {array} \right]^t 
$$
which has full rank and hence $\mathbf{N}$ is a $\vec w$-minimal kernel basis.  If 
$$
\mathbf{N}_1 = \left[ \begin {array}{ccc} 1&2\,{x}^{2}+1&2\,{x}^{4}
\\ \noalign{\medskip}1&2\,{x}^{2}+2&2\,{x}^{4}+{x}^{2}
\\ \noalign{\medskip}1&2\,{x}^{2}+1&2\,{x}^{4}+1\end {array} \right] 
~~ \mbox{  and } ~~ 
\mathbf{N}_2 = \left[ \begin {array}{ccc} x-1&x-1&-2\,x+2\\ \noalign{\medskip}1&x+2&
-2\\ \noalign{\medskip}-3\,{x}^{3}&{x}^{5}-3\,{x}^{3}-x&{x}^{7}-{x}^{3
}+1\end {array} \right]
$$
denote the first and last three rows of $\mathbf{N}$, respectively, then $\mathbf{N}_1$
is unimodular, $\mathbf{F} \cdot \mathbf{N}_1 = \mathbf{N_2}$ and $\mathbf{N_2}$ has the
same row degrees as the Hermite form. Then
$$
U =  \left[ \begin {array}{crc} 1&-1&2\\ \noalign{\medskip}0&1&0
\\ \noalign{\medskip}0&0&1\end {array} \right]
$$ 
is the matrix which converts the leading coefficient matrix of $x^{ - (1,1,7)} \cdot \mathbf{N}_2$ into column echelon form. Therefore $\mathbf{U} = \mathbf{N}_1 \cdot U$  is unimodular and 
$$
\mathbf{F} \cdot \mathbf{U} = \left[ \begin {array}{ccc} x-1&0&0\\ \noalign{\medskip}1&x+1&0
\\ \noalign{\medskip}-3\,{x}^{3}&{x}^{5}-x&{x}^{7}+1\end {array}
 \right]
$$ 
is in Hermite normal form. \qed
\end{exmp}

\subsection{Reducing the degrees and the shift}\label{subsec:2}

Although the Hermite normal form $\mathbf{H}$ of $\mathbf{F}$ can be computed from a $\left[-\vec{u},-\vec{s}\right]$-minimal kernel basis of $\left[\mathbf{F}, - I\right]$, a major problem here is that $s_{\max}$ can be very large and hence is not efficient enough for our needs. In this subsection we show how to convert our normal form computation into one where the bounds are smaller. 

Since we know the row degrees of $\mathbf{H}$ we can expand each of the high degree rows of $\mathbf{H}$ into multiple rows each having lower degrees. This allows us to compute an alternative matrix $\mathbf{H}'$ having lower degrees, but with a higher row dimension that is still in $O(n)$, with $\mathbf{H}$ easily determined from this new $\mathbf{H}'$.
Such an approach was used in the Hermite form computation algorithm  in \cite{GS2011} and also in the case of the order basis algorithms from \cite{za2009,ZL2012}. Our task here is in fact easier as we already know the row degrees of $\mathbf{H}$.

For each entry $s_{i}$ of the shift $\vec{s}$, let $q_{i}$ and $r_{i}$ be the quotient and remainder of $s_{i}$ divided by $d_{\max}$. Any polynomial $h$ of degree at most $s_i$ can be written as
$$
h = h_0 + h_1 x^{s_i-q_i d_{\max}} + \cdots +  h_{q_i} x^{s_i - d_{\max}}
$$
with each of the $q_i+1$ components, $h_i$, having degree at most $d_{\max}$. Similarly the $i$th column of $\mathbf{H}$ can be written in terms of vector components each bounded in degree by $d_{\max}$. 

In terms of matrix identities let
\[
\tilde{\mathbf{E}}_{i}=\left[e_{i},x^{s_{i}-q_{i}d_{\max}}e_{i},\dots,x^{s_{i}-d_{\max}}e_{i}\right] 
\]
where $e_{i}$ is the $i$th column of the $n \times n$ identity matrix. For each $i$ $\tilde{\mathbf{E}}_{i}$ has
$q_{i}+1$ entries and combines into a single matrix 
\begin{eqnarray*}
\mathbf{E} & = & [\tilde{\mathbf{E}}_1,\dots,\tilde{\mathbf{E}}_{n}]\\ & & \\
 & = & \left[\begin{array}{cccc|cc|cccccc}
1 & x^{s_{1}-q_{1}d_{\max}} & \cdots  & x^{s_{1}-d_{\max}} & & \\
\hline 
 &  &  &  &   \ddots & & \\
 &  &  &  &   & \ddots\\
\hline 
 &  &  &  &  & &  & 1 & x^{s_{n}-q_{n}d_{\max}} & \cdots  & x^{s_{n}-d_{\max}}
\end{array}\right]_{n\times\bar{n}}
\end{eqnarray*}
with column dimension $\bar{n}= n+\sum_{i=1}^{n}q_{i}.$ 
Then there exists a matrix polynomial $\mathbf{H_E} \in \mathbb{K} [x]^{\bar{n} \times n}$ such that
$\mathbf{E} \cdot \mathbf{H_E}  = \mathbf{H}$. The shift in this case is given by
\begin{eqnarray*}
\vec{s^*} & = & [\vec{s^*}_{1},\dots,\vec{s^*}_{n}]~~~ \mbox{ where } ~~ 
\vec{s^*}_i  =   [r_i , d_{\max}, \ldots , d_{\max}] ~~\mbox{  with  } {q_i+1} \mbox{  components }.
\end{eqnarray*}

Since $\mathbf{F}$ is column reduced we have that $\sum_{i=1}^{n} s_i = \sum_{i=1}^{n} d_i \le n d_{\max}$ and hence the column dimension $\bar{n}$ satisfies 
\[
\bar{n}=n+\sum_{i=1}^{n}q_{i}\le n+\sum_{i=1}^{n} \frac{s_{i}}{d_{\max}}\le 2n .
\]

Rather than recovering $\mathbf{H}$ from a  $\left[-\vec{u},-\vec{s}\right]$-minimal kernel basis for $[ \mathbf{F}, -\mathbf{I} ]$ the following lemma implies that we can now recover the Hermite form from a $\left[-\vec{u},-\vec{s^*}\right]$-minimal
kernel basis of $\left[\mathbf{F},-\mathbf{E}\right]$.
\begin{lem}\label{lem:expandH}
Let $\mathbf{N}$ be a $\left(\left[\mathbf{F},-\mathbf{E}\right],\left[-\vec{u},-\vec{s^*}\right]\right)$-kernel
basis, partitioned as $\mathbf{N}=\begin{bmatrix}\mathbf{N}_u\\
\mathbf{N}_d
\end{bmatrix}$ with $\mathbf{N}_d$ of dimension $\bar{n}\times\bar{n}$.
Let $\bar{\mathbf{N}}$ be the submatrix consisting of the columns
of $\mathbf{N}_d$ whose $-\vec{s^*}$-column degrees are bounded by
$0$. Then $\mathbf{H}$ is a column basis of $\mathbf{E} \cdot \bar{\mathbf{N}}$ having $-\vec{s}$-column degrees $0$.
\end{lem}
\begin{proof}
Let $\mathbf{U}$ be the unimodular matrix satisfying $\mathbf{F} \cdot \mathbf{U}=\mathbf{H}$ and
$\mathbf{H}_{\mathbf{E}}$ the matrix $\mathbf{H}$ expanded according
to $\mathbf{E}$. Then $\mathbf{H}_{\mathbf{E}}$ has row degrees $\vec{s^*}$ and $\mathbf{E} \cdot \mathbf{H}_{\mathbf{E}}=\mathbf{H}$. Set 
\[
\mathbf{M}=\begin{bmatrix}\mathbf{U} & 0\\
\mathbf{H}_{\mathbf{E}} & \mathbf{N}_{\mathbf{E}}
\end{bmatrix}
\]
where $\mathbf{N}_{\mathbf{E}}$ is a $\left(\mathbf{E},-\vec{s^*}\right)$-kernel
basis.
Then $\mathbf{M}$ is a $\left[-\vec{u},-\vec{s^*}\right]$ kernel basis for 
$\left[\mathbf{F},-\mathbf{E}\right]$.

Let $\mathbf{N}=\begin{bmatrix}\mathbf{N}_u\\
\mathbf{N}_d
\end{bmatrix}$ be a $\left(\left[\mathbf{F},-\mathbf{E}\right],\left[-\vec{u},-\vec{s^*}\right]\right)$-kernel
basis. Then the matrices $\mathbf{M}_{0}$ and $\mathbf{N}_{0}$ consisting
of the columns from $\mathbf{M}$ and $\mathbf{N}$, respectively,
whose $\left[-\vec{u},-\vec{s^*}\right]$-column degrees are bounded
by $0$, are unimodularly equivalent, that is, $\mathbf{M}_{0} \cdot \mathbf{V}=\mathbf{N}_{0}$
for some unimodular matrix $\mathbf{V}$. As a result, the matrices
$\bar{\mathbf{M}}$ and $\bar{\mathbf{N}}$ consisting of the
bottom $\bar{n}$ rows of $\mathbf{M}_{0}$ and $\mathbf{N}_{0}$
respectively, also satisfies $\bar{\mathbf{M}}\mathbf{V}=\bar{\mathbf{N}}$.
Thus $\mathbf{E} \cdot \bar{\mathbf{M}} \mathbf{V}=\mathbf{E}\bar{\mathbf{N}} $,
with $\cdeg_{-\vec{s}}\mathbf{E} \cdot \bar{\mathbf{M}} \le 0$ and $\cdeg_{-\vec{s}}\mathbf{E}\bar{\mathbf{N}}\le0$,
since $\cdeg_{-\vec{s^*}}\bar{\mathbf{M}}\le0$ and $\cdeg_{-\vec{s^*}}\bar{\mathbf{N}}\le0$.
Let $\bar{\mathbf{M}}=\left[\mathbf{H}_{\mathbf{E}},\mathbf{N}'\right]$,
where $\mathbf{N}'$ consists of the columns of $\mathbf{N}_{E}$
with $-\vec{s^*}$-column degrees bounded by $0$. Then 
$$
\mathbf{E} \cdot \bar{\mathbf{M}} \mathbf{U}_{0}=\mathbf{E}\left[\mathbf{H}_{\mathbf{E}},\mathbf{N}'\right]\mathbf{V}=\left[\mathbf{H},0\right]\mathbf{V}=\mathbf{E}\bar{\mathbf{N}}.
$$
Since $\mathbf{H}$ is $-\vec{s}$-column reduced and has $-\vec{s}$-column
degrees $0$, the nonzero columns of $\mathbf{E} \cdot \bar{\mathbf{N}}$
must have $-\vec{s}$-column degrees no less than $0$, so their
$-\vec{s}$-column degrees are equal to $0$.
\end{proof}

\begin{rem}\label{rmk1}
Note that the nonzero columns of $\mathbf{E} \cdot \bar{\mathbf{N}}$ having
$-\vec{s}$-column degrees $0$ allows us to recover the Hermite form $\mathbf{H}$ of $\mathbf{F}$
from $\mathbf{E} \cdot \bar{\mathbf{N}}$ using \prettyref{lem:recoverH}
\end{rem}

\begin{exmp}\label{example4}
Let $\mathbf{F}$ be the $3 \times 3$ matrix over $\mathbb{Z}_7[x]$ from Example \ref{example3}. In this case $d_{\max} = 3$, $q_1 =  q_2 = 0$, $r_1 = r_2 = 1$ and $q_3 = 2, r_3 =1$ so that
$$
\mathbf{E} = \left[ \begin{array}{ccccc} 1 & 0 & 0 & 0 & 0 \\
0 & 1 & 0 & 0 & 0 \\
0 & 0 & 1 & x & x^4 \end{array} \right]
$$
with shift $\vec{ s^*} = (1, 1, 1, 3, 3)$. Then a $[ - \vec{u}, - \vec{s^*} ]$-minimal kernel basis for $[ \mathbf{F}, - \mathbf{E} ]$ is given by the $8 \times 5$ matrix
$$
\mathbf{N} = 
\left[ 
\begin{array}{ccc|ccccc} 1&1&1&x-1&1&0& -3 x^2 & 0\\ 
\noalign{\medskip}2\,{x}^{2}+1&2\,{x}^{2}+2&2\,{x}^{2}+1&x-1&x+2&0 & -3\,{x}^{2}-1 & x \\ 
\noalign{\medskip}2\,{x}^{4}&2\,{x}^{4}+{x}^{2}&2
\,{x}^{4}+1&-2\,x+2&-2&1 & - x^2 & x^3 
\\
0 & 0 & 0 & 0 & ~0 & -x & 1 & 0 \\
0 & 0 & 0 & 0 & ~0 & 0 & -x^3 & 1 
\end{array} \right]^t ~.
$$
The matrix $\bar{\mathbf{N}}$, the bottom $5$ rows of $\mathbf{N}$ having shifted degrees less than
$0$, then gives
$$
\mathbf{E} \cdot \bar{\mathbf{N}} = 
 \left[ \begin {array}{ccccc} x-1&x-1&-2\,x+2 & 0 & 0 \\ \noalign{\medskip}1&x+2&
-2 & 0 & 0 \\ \noalign{\medskip}-3\,{x}^{3}&{x}^{5}-3\,{x}^{3}-x&{x}^{7}-{x}^{3
}+1 & 0 & 0\end {array} \right]
$$
with the column space consisting of the first $3$ columns. As in Example \ref{example3} the Hermite form is then determined by reduction of the shifted leading coefficient matrix of this column space into column echelon form.
\qed
\end{exmp}

\subsection{Efficient Remainder Computation}

The problem of computing a $\left(\left[\mathbf{F},I\right],\left[-\vec{u},-\vec{s}\right]\right)$-kernel
basis from Subsection \ref{subsec:1} has now been reduced to computing a $\left(\left[\mathbf{F},\mathbf{E}\right],\left[-\vec{u},-\vec{s^*}\right]\right)$-kernel
basis in Subsection \ref{subsec:2}. However, the degree of $\mathbf{E}$ and the shift $\vec{u}$
are still too large to guarantee efficient computation. In order to lower these further,
we can reduce $\mathbf{E}$ against $\mathbf{F}$ in a way similar to \cite{GS2011}. 
\begin{lem}\label{lem:reduceToRemainder} 
Suppose
$$\mathbf{E} ~ = ~\mathbf{F} \cdot \mathbf{Q} + \mathbf{R}$$ where $\mathbf{Q}$ and $\mathbf{R}$ are matrix polynomials and where $\mathbf{R}$ has degree less than $d_{\max}$. Let $\vec{u^*}=\left[2d_{\max},\dots,2d_{\max}\right]\in\mathbb{Z}^{n}$
and  $\mathbf{N}=\begin{bmatrix}{\mathbf{N}_u}\\
\mathbf{N}_d
\end{bmatrix}$ be a $\left(\left[\mathbf{F},-\mathbf{R}\right],\left[-\vec{u^*},-\vec{s^*}\right]\right)$-kernel
basis where the block $\mathbf{N}_d$ has dimension $n \times n$. Let
$\bar{\mathbf{N}}$ be the matrix consisting of the columns
of $\mathbf{N}_d$ whose $-\vec{s^*}$-column degrees are bounded by
$0$. Then $\mathbf{H}$ is a column basis of $\mathbf{E} \cdot \bar{\mathbf{N}}$
and the nonzero columns of $\mathbf{E} \cdot \bar{\mathbf{N}}$ have
$-\vec{s}$-column degrees $0$. 
\end{lem}
\begin{proof}
Let $\mathbf{N}=\begin{bmatrix}\mathbf{U} & 0\\
\mathbf{H}_{\mathbf{E}} & \mathbf{N}_{\mathbf{E}}
\end{bmatrix}$ be the kernel basis of $\left[\mathbf{F},-\mathbf{E}\right]$ constructed in \prettyref{lem:expandH}. Then 
\[
\left[\mathbf{F},-\mathbf{R}\right]=\left[\mathbf{F},-\mathbf{E}\right]\begin{bmatrix} \mathbf{I_n} & \mathbf{Q}\\
0 & \mathbf{I_{\bar{n}}}
\end{bmatrix}
\]
implies that we can construct a kernel basis for $\left[\mathbf{F},-\mathbf{R}\right]$ via
\[
\mathbf{\hat{N}}=\begin{bmatrix} \mathbf{I_n} & -\mathbf{Q}\\
0 & \mathbf{I_{\bar{n}}}
\end{bmatrix}\mathbf{N}=\begin{bmatrix}\mathbf{U}-\mathbf{Q}\mathbf{H}_{\mathbf{E}} & -\mathbf{Q}\mathbf{N}_{E}\\
\mathbf{H}_{\mathbf{E}} & \mathbf{N}_{\mathbf{E}}
\end{bmatrix}.
\]
The proof then follows the same argument as given in Lemma \ref{lem:expandH}.
\end{proof}

As noted in Remark \ref{rmk1} the fact that the nonzero columns of $\mathbf{E} \cdot \bar{\mathbf{N}}$ have
$-\vec{s}$-column degrees $0$ allows us to recover $\mathbf{H}$
from $\mathbf{E} \cdot \bar{\mathbf{N}}$ using column echelon reduction of a constant matrix as noted in \prettyref{lem:recoverH}.

A $\left(\left[\mathbf{F},-\mathbf{R}\right],\left[-\bar{u},-\vec{s^*}\right]\right)$-kernel
basis from \prettyref{lem:reduceToRemainder} can now be computed and can then be used to recover the Hermite normal form. A big question remaining, however, is how to efficiently compute the
remainder $\mathbf{R}$ from $\mathbf{E}$ and $\mathbf{F}$. For this, we follow the approach used
for computing matrix gcds found in \cite{BL2000} and use  $\mathbf{F^*}$, the matrix polynomial having the reverse coefficients of $\mathbf{F}$:
\[
\mathbf{F^*}(x)  = \mathbf{F}(1/x) \cdot x^{\vec{d}} =  \mathbf{F}(1/x)\begin{bmatrix}x^{d_{1}}\\
 & \ddots\\
 &  & x^{d_{n}}
\end{bmatrix}.
\]
Efficient computation can then be done using the series expansions of the inverse of $\mathbf{F^*}$ as in the proof of Lemma 3.4 from \cite{Giorgi2003} and making use of the higher-order lifting algorithm from \cite{storjohann:2003}. Since $\mathbf{F}$ is assumed to be column reduced, $\deg\det\mathbf{F}=\sum\vec{d}$, and so the lowest order term of $\mathbf{F^*}$ is nonsingular. Therefore $x$ is not a factor of $\det\mathbf{F^*}$, which means the series expansion of $( \mathbf{F^*} )^{-1}$ always exists. This also implies that the $x^{d}$-adic lifting  algorithm from \cite{storjohann:2003} becomes deterministic. 

The division with remainder also takes advantage of the special structure of $\mathbf{E}$. In particular, its higher degree elements are all of the form $x^k e_i$ for $e_i$ a column of the identity. Consider now how the series expansion of $\left(\mathbf{F^*}\right)^{-1}$ gives a remainder of $x^{k}I$ divided by $\mathbf{F}$.
\begin{lem}\label{lem:remainder}
For any integer $k\ge d_{\max}$, let $\mathbf{G_k^*} = ( \mathbf{F^*})^{-1} \mod x^{k - \vec{d}}$ where $({\mathbf{F^*})^{-1}} \mod x^{k-\vec{d}}$ denotes the $i$-th row of the series expansion of ${(\mathbf{F^*})^{-1}}$ $\mod$ $x^{k-d_{i}}$ for each row $i$. Then 
\begin{equation}
I = \mathbf{F^*} \cdot \mathbf{G_k^*} + x^{k- d_{\max}}\mathbf{C_k^*}\label{eq:division}
\end{equation}
where 
$\mathbf{C_k} = x^{{d}_{\max}}\mathbf{C_k^*}(1/x) $ is a matrix polynomial having degree
less than ${d}_{\max}$ and satisfying $$x^{k}I=\mathbf{F} \cdot \mathbf{G_k} +\mathbf{C_k}$$
with $\mathbf{G_k} = x^{k - \vec{d}} \cdot \mathbf{G_k^*}(1/x)$ a matrix polynomial.
\end{lem}
\begin{proof}
Replacing $x$ by $1/x$ in equation (\ref{eq:division}) and multiplying both sides by $x^k$ gives
\begin{eqnarray*}
x^k I & = &  x^k \mathbf{F^*}(1/x) \cdot \mathbf{G_k^*}(1/x) + x^{ d_{\max}}\mathbf{C_k^*}(1/x) \\
 & = & \mathbf{F}(x) \cdot x^{k - \vec{d}} \cdot \mathbf{G_k^*}(1/x) + x^{ d_{\max}}\mathbf{C_k^*}(1/x) \\
 & = & \mathbf{F} \cdot \mathbf{G_k} + \mathbf{C_k} ~.
\end{eqnarray*}
The definition of $\mathbf{G_k^*}$ implies that $\mathbf{G_k}$ is a matrix polynomial. Therefore $\mathbf{C_k}$ is also a matrix polynomial and, by its definition, has degree less than $d_{\max}$. 
\end{proof}

Lemma \ref{lem:remainder} implies that $x^k \cdot e_j$ is determined by the $j$-th column of
$\mathbf{G_k}$ and $\mathbf{C_k}$ allowing us to construct the columns of the quotient and
remainder for the corresponding columns when $\mathbf{E}$ is divided by $\mathbf{F}$.

\prettyref{lem:remainder} shows how the series expansion 
$$
({\mathbf{F^*}})^{-1} = \mathbf{G} = G_0 + G_1 x + G_2 x^2 + \cdots 
$$
can be used to compute a remainder of $x^{k}I$ divided by $\mathbf{F}$
for any $k \ge d_{\max}$. Similarly, the $i$-th column ${\mathbf{G}}_{i} = {\mathbf{G}}e_{i}$
of ${\mathbf{G}}$ allows us to compute a remainder $\mathbf{r}$ of $x^{k}e_{i}$ divided by $\mathbf{F}$, with $\deg\mathbf{r}<d_{\max}$. Note that the degrees of columns corresponding to $e_{i}$ are bounded by $s_{i}$, so we need to compute the series expansions ${\mathbf{G}}_{i}$
to at least order $s_{i}$. Let us now look to see how these series expansions can be computed efficiently.

\begin{lem}
Let $\mathbf{G} = (\mathbf{F^*})^{-1}$. Then computing the series expansions ${\mathbf{G}}_{i}$ to order $s_{i}$ for all $i$'s where  $s_{i}\ge d_{\max}$ can be done with a cost of \textup{$O^{\sim}\left(n^{\omega}d_{\max}\right)$
field operations.}
\end{lem}
\begin{proof}
Let us assume,  without any loss of generality, that the columns of $\mathbf{F}$ and the corresponding entries of $\vec{s}=\left[s_{1},\dots,s_{n}\right]$ are arranged so that they are in increasing order. We can then separate $\vec{s}$ into $\left\lceil \log n\right\rceil +1$ disjoint
lists  $\vec{s}_{{j}^{\left(0\right)}},\vec{s}_{{j}^{\left(1\right)}},\vec{s}_{{j}^{\left(2\right)}},\dots,\vec{s}_{{j}^{\left(\left\lceil \log n\right\rceil \right)}}$
with entries in the ranges $[0,d_{\max})$, $[d_{\max},2d_{\max})$, ..., $[2^{\left\lceil \log n\right\rceil -2}d_{\max},2^{\left\lceil \log n\right\rceil -1}d_{\max})$,
$[2^{\left\lceil \log n\right\rceil -1}d_{\max},nd_{\max}]$ respectively, where
each ${j}^{(i)}$ consists of a list of indices of the entries of
$\vec{s}$ that belong to $\vec{s}_{{j}^{\left(i\right)}}$. Notice
that ${j}^{\left(i\right)}$ has at most $n/2^{i-1}$ entries,
since otherwise the sum of the entries of $\vec{s}_{{j}^{(i)}}$ would
exceed $\sum\vec{s} \ge nd_{\max}$. We then compute series expansions ${\mathbf{G}}_{{j}^{\left(1\right)}},{\mathbf{G}}_{{j}^{\left(2\right)}},\dots,{\mathbf{G}}_{{j}^{\left(\left\lceil \log n\right\rceil \right)}}$
separately, to order $2d_{\max},4d_{\max},\dots,2^{\left\lceil \log n\right\rceil -1}d_{\max}/2,nd_{\max},$
respectively, where again ${\mathbf{G}}_{{j}^{(i)}}$ consists
of the columns of ${\mathbf{G}}$ that are indexed by the entries
in ${j}^{\left(i\right)}$. These computations are done efficiently using the higher order series solution algorithm from \cite{storjohann:2003}. For ${\mathbf{G}}{}_{{j}^{\left(i\right)}}$,
there are at most $n/2^{i-1}$ columns, so computing the series expansion
to order $2^{i}d_{\max}$ costs $O^{\sim}\left(n^{\omega}d_{\max}\right)$. Doing this for $i$ from $1$ to $\left\lceil \log n\right\rceil $ then costs $O^{\sim}\left(n^{\omega}d_{\max}\right)$ field operations. 
\end{proof}
With the series expansions computed, we can compute a remainder
$\mathbf{R}$ of $\mathbf{E}$ divided by $\mathbf{F}$.
\begin{lem}
A remainder $\mathbf{R}$ of $\mathbf{E}$ divided by $\mathbf{F}$,
where $\deg\mathbf{R}<d_{\max}$, can be computed with a cost of $O^{\sim}\left(n^{\omega}d_{\max}\right)$
field operations.
\end{lem}
\begin{proof}
The remainder $\mathbf{r}$ of $x^{k}e_{i}$ divided by $\mathbf{F}$
can be obtained by 
\[
\left(e_{i}-\mathbf{F^{*}}\left({\mathbf{G}}_{i}\mod x^{k-\vec{d}}\right)\right)/x^{k-d_{\max}} ~
\mbox{ where } \mathbf{G} = ({\mathbf{F^*}})^{-1}.
\]
Note that only the terms from ${\mathbf{G}}_{i}$ with degrees in the range $[k-2d_{\max},k)$ are needed for this computation, which means we are just multiplying $\mathbf{F^*}$ with a polynomial vector with degree bounded by $2d_{\max}$. In order to make the multiplication more efficient,
we can compute all the remainder vectors at once. Since there at most $n$ columns with degrees no less than $d_{\max}$, the cost is just the
multiplication of matrices of dimension $n$ and degrees bounded by
$2d_{\max}$, which costs $O^{\sim}\left(n^{\omega}d_{\max}\right)$ field operations.
\end{proof}

\begin{exmp}\label{example5}
Continuing with Example \ref{example4} we see that only column $5$ of $\mathbf{E}$, which is $x^4 \cdot e_3$ needs to be reduced. In this case $\mathbf{G_4^*} = (\mathbf{F^{*}})^{-1} ~ \mbox{ mod } x = {\tiny \left[ \begin {array}{ccc} 0&0&2\\ \noalign{\medskip}3&0&2
\\ \noalign{\medskip}0&-1&2\end {array} \right]} 
 $ and we have that 
$$
\mathbf{G_4} = \left[ \begin {array}{ccc} 0&0&2\,x\\ \noalign{\medskip}3\,x&0&2\,x
\\ \noalign{\medskip}0&-x&2\,x\end {array} \right]
\mbox{ and } ~~ \mathbf{C_4} = \left[ \begin {array}{ccc} {x}^{2}&-2\,x+2&-2\,x+2
\\ \noalign{\medskip}-{x}^{2}-3\,x-3&-{x}^{2}-2&-2
\\ \noalign{\medskip}3\,x&1&0\end {array} \right] .
$$
Therefore we construct $\mathbf{Q}$ as 
$$
\mathbf{Q} = \left[ \begin{array}{ccccc} 0 & 0 & 0 & 0 & 2 x \\
0 & 0 & 0 & 0 & 2x \\
0 & 0 & 0 & 0 & 2x  
\end{array} \right]
$$
and then multiply the $\mathbf{N}$ matrix from Example \ref{example4} by
$$
\left[ \begin{array}{cc} \mathbf{I_3} & - \mathbf{Q} \\ 0 & \mathbf{I_5} \end{array} \right]
$$
which then gives a minimal kernel basis for $[\mathbf{F} ,  - \mathbf{R} ]$ as
$$
\mathbf{\hat{N}} = 
\left[ 
\begin{array}{ccc|ccccc} 1&1&1&x-1&1&0& -3 x^2 & 0\\ 
\noalign{\medskip}1&2\,{x}^{2}+2&2\,{x}^{2}+1&x-1&x+2&0 & -3\,{x}^{2}-1 & x \\ 
\noalign{\medskip}0&2\,{x}^{4}+{x}^{2}&2
\,{x}^{4}+1&-2\,x+2&-2&1 & - x^2 & x^3 
\\
0 & 0 & 0 & 0 & ~0 & -x & 1 & 0 \\
-2 x & 0 & 0 & 0 & ~0 & 0 & -x^3 & 1 
\end{array} \right]^t ~.
$$
The Hermite form for $\mathbf{F}$ is then determined using the bottom $5$ rows of
$\mathbf{\hat{N}}$ as in Example \ref{example4}.
\qed
\end{exmp}

With the remainder $\mathbf{R}$ computed, we can now  compute a $\left(\left[\mathbf{F},-\mathbf{R}\right],\left[-\bar{u},-\vec{s^*}\right]\right)$-kernel
basis that can be used to recover the Hermite normal form using \prettyref{lem:reduceToRemainder}.
This in turn gives us our Hermite form.
\begin{thm}
A Hermite normal form of $\mathbf{F}$ can be computed deterministically
with a cost of \textup{$O^{\sim}\left(n^{\omega}d_{\max}\right)$ field operations.}
\end{thm}

\section{Conclusion\label{sec:Future-Research}}

In this paper we have given a new, deterministic algorithm for computing the Hermite normal form of a nonsingular matrix polynomial. Our method relied on the efficient, deterministic computation of the diagonal elements of the Hermite form. In terms of future research we are interested in finding efficient, deterministic algorithms for other forms such as the Popov normal form. We also expect that our methods can be improved so that normal form computation can be done in cost $O^{\sim}\left(n^{\omega-1}s\right)$, that is with maximal degree replaced by average degree. 
In addition we are interested in fast, normal form algorithms where then entries are differential operators rather than polynomials. Such algorithms are useful for reducing systems of linear differential equations to solveable systems \cite{barkatou:2013}. 

\bibliographystyle{plainnat}

\end{document}